\documentclass[12pt, journal, draftclsnofoot, onecolumn]{IEEEtran}

\usepackage[utf8]{inputenc} 
\usepackage[T1]{fontenc}
\usepackage{url}
\usepackage{ifthen}
\usepackage[cmex10]{amsmath} % Use the [cmex10] option to ensure complicance
\IEEEoverridecommandlockouts

\usepackage{color}

\usepackage{graphicx,tabularx,array,amsmath,amsthm,thmtools}
\usepackage{enumitem}
\usepackage{algorithm}
\usepackage{algorithmic}
\usepackage{mathtools}

\usepackage{amsfonts}
\usepackage{bm}
\usepackage{bbm}
\usepackage{makecell}
\usepackage{multirow}
 \usepackage{amssymb}
\usepackage{txfonts}
\usepackage[T1]{fontenc}
\usepackage{tikz}
\usepackage[scr=dutchcal]{mathalfa}

\usepackage{caption}
\usepackage{subcaption}

\usepackage{cite}

\newtheorem{Theorem}{Theorem}

\newtheorem{Lemma}{Lemma}

\newtheorem{Definition}{Definition}

\theoremstyle{definition}
\newtheorem{Example}{Example}

\DeclareMathOperator*{\argmax}{argmax}

\begin{document}

\author{Farhad Shirani, Shahram Shahsavari, and Elza Erkip \\
Electrical and Computer Engineering Department \\
New York University, NY, USA \\\date{} }

\title{On the Rates of Convergence in Learning of Optimal Temporally Fair Schedulers}


\maketitle
\begin{abstract}
%Multi-user scheduling is one of the enabling techniques in the next generation of wireless networks.
 Multi-user schedulers are designed to achieve optimal average system utility (e.g. throughput) subject to a set of fairness criteria. In this work, scheduling under temporal fairness constraints is considered.
%, where each user is required to be activated for at least a predefined fraction of the resource blocks over a given scheduling window. 
Prior works have shown that a class of scheduling strategies called threshold based strategies (TBSs) achieve optimal system utility under temporal fairness constraints. %A TBS assigns a threshold value to each user in the network. The subset of users to be activated at each resource block is a function of the resulting throughput --- which varies stochastically at different resource blocks --- and the assigned thresholds. The value of the optimal thresholds depend on the channel statistics.  
%One of the main difficulties in designing TBSs is the fact that
%However, the scheduler does not typically have prior knowledge of the channel statistics,
% and hence does not know the optimal value of the thresholds beforehand. 
%and instead, it
The optimal TBS thresholds are determined as a function of the channel statistics.
 In order to provide performance guarantees for TBSs in practical scenarios --- where the scheduler learns the optimal thresholds based on the empirical observations of the channel realizations --- it is  necessary to evaluate the rates of convergence of TBS thresholds to the optimal value.   
In this work, these rates of convergence and the effect on the resulting system utility are investigated. It is shown that the best estimate of the threshold vector is at least $\omega(\frac{1}{\sqrt{t}})$ away from the optimal value, where $t$ is the number of observations of the independent and identically distributed channel realizations. Furthermore, it is shown that under long-term fairness constraints, the scheduler may achieve an average utility that is higher than the optimal long-term utility by violating the fairness criteria for a long initial period. Consequently, the resulting system utility may converge to its optimal long-term value from above. The results are verified by providing simulations of practical scheduling scenarios. 

\let\thefootnote\relax\footnotetext{This work is supported by NYU WIRELESS Industrial Affiliates and National Science Foundation grants EARS-1547332 and NeTS-1527750. This work was done when S. Shahsavari was at NYU.}

\end{abstract}

\section{Introduction}

Network information theory studies the set of optimal achievable rates in a cellular system for a \textit{given set of users}. Scheduling strategies determine \textit{which users} are activated by the base station (BS) at each transmission block \cite{ kulkarni2003opportunistic,shahsavari2019general,liu-jsac,issariyakul2004throughput,rosenberg-short-term}. The choice is made in accordance with the users' fairness demands (e.g. temporal demands), the resulting system utility (e.g. sum-rate), and the computational complexity constraints at the transmitter and receiver which limit the number of users which can be activated simultaneously. Once the choice of active users is made, physical layer techniques  are used for communication over the resulting channel \cite{ahlswede1973multi,yu2004sum}. 

% Opportunistic scheduling strategies are designed to achieve maximum system utility by exploiting the channel state information subject to a set of fairness criteria \cite{ kulkarni2003opportunistic,shahsavari2019general,liu-jsac,issariyakul2004throughput,rosenberg-short-term}. Various criteria have been proposed to model and evaluate fairness. These criteria are categorized as temporal  \cite{liu-jsac,kulkarni2003opportunistic, shahsavari2019general} utilitarian \cite{liu-elsevier,zhang2008opportunistic}, and proportional \cite{kelly1998rate,viswanath2002opportunistic} fairness criteria. 

%In this work, we  consider scheduling under temporal fairness , where the fraction of the resource blocks in which each user is activated is required to be within some predetermined upper and lower bounds. 

Temporally fair schedulers provide each user with a minimum temporal share in order to control the average delay, and restrict the maximum power drain of users by placing upper-bounds on their temporal shares \cite{kulkarni2003opportunistic,issariyakul2004throughput,liu-jsac, shahsavari2019general}.
%From the perspective of the network provider, an additional upside of temporally fair schedulers is that users with low channel quality do not reduce the average system throughput as severely as in utilitarian fair schedulers \cite{asadi2013survey}.
In \cite{shahsavari2019general,arxiv_PIMRC,shiranishortterm}, we consider the user scheduling problem for non-orthogonal multiple access (NOMA) and full-duplex (FD) systems under temporal fairness constraints, and show that optimal average system utility is achieved using a class of scheduling strategies called \textit{Threshold Based Strategies} (TBSs). A TBS assigns real-valued thresholds to each of the users in the network. At each resource block a subset of users is activated based on the resulting system utility --- which is calculated based on the channel realizations in that time-slot --- and the thresholds assigned to each of the users. The optimal thresholds assigned to each user in the TBS are functions of the users' channel statistics. In practice, the scheduler does not typically have prior knowledge of the users' channel statistics. Rather, it gains an empirical estimate of the statistics using its observations of the prior channel realizations obtained through channel estimation and feedback techniques. Consequently, the scheduler cannot calculate the optimal thresholds prior to the start of communication, and it has to update the thresholds in an online fashion as it gains a more accurate empirical estimate of the channel statistics by accumulating observations of independent and identically distributed channel realizations \cite{liu-jsac,shahsavari2019general}. Consequently, in order to provide performance guarantees for TBSs in terms of achievable system utility, one needs to evaluate the rate of convergence (RoC) of the thresholds to  the optimal thresholds. 

In this paper, we investigate the RoC of the thresholds used in the TBS strategy to the optimal threshold values as well as the RoC of the resulting average system utility to the optimal utility. We show that the threshold values converge with rate at most\footnote{We write $f(x)=O(g(x))$ if $\lim_{n\to \infty}\frac{f(x)}{g(x)}< \infty$.} $O(\frac{1}{\sqrt{t}})$, where $t$ is the number of the prior empirical observations of the users' channel statistics. 
%We argue that this upper bound on the RoC is in agreement with the well-known lower bounds on the RoC of stochastic approximation algorithms \cite{polyak1992acceleration}, which suggests that the bound is tight. 
This bound on the RoC of the TBS thresholds can be used to derive gaurantees in terms of scheduler's short-term fairness and average system utility. 
We also investigate the RoC for the average system utility, and show the existence of a scheduling strategy whose utility converges to the optimal long-term utility from above. Loosely speaking, this is achieved by violating the temporal fairness constraints for a long initial period while accumulating additional system utility. After this initial period, having gained an accurate estimate of the channel statistics, the scheduler operates at near-optimal utility while satisfying the fairness constraints. 
We verify the results by providing simulations of practical scenarios. 

%The rest of the paper is organized as follows: Section \ref{Sec:Not} provides the notation used in the paper. Section \ref{sec:Prel} explains the necessary background and problem formulation. 
%Section \ref{sec:thresh} derives RoC for the thresholds in the TBS and the resulting system utility. Section \ref{Sec:sim} provides several simulations of practical scheduling scenarios. Section \ref{Sec:Conc} concludes the paper.  

\textit{Notation:} The set of numbers $\{1,2,\cdots, n\}, n\in \mathbb{N}$ is represented by $[n]$.  
The vector $(x_1,x_2,\cdots, x_n)$ is written as $x^n$. The $m\times t$ matrix $[g_{i,j}]_{i\in [m], j\in [t]}$ is denoted by $g^{m\times t}$.
% For a random variable $X$, the corresponding probability space is $(\mathcal{X}, \mathbf{F}_{X}, P_X)$, where $\mathbf{F}_X$ is the underlying $\sigma$-field. 
 For an event $\mathcal{A}$, the random variable $\mathbbm{1}_{\mathcal{A}}$ is the indicator function. For the continuous random variable $X$ whose probability density function (PDF) depends on the model parameter $\theta$, we write the PDF as $f_X(x;\theta), x\in \mathbb{R}$. 
  %The notation $\floor{x}$ ($\ceil{x}$) is used to represent the closest integer smaller (bigger) than $x$. 

\section{Preliminaries}
\label{sec:Prel}
We focus on opportunistic NOMA scheduling under temporal fairness constraints which was formulated in \cite{shahsavari2019general}. The main results may be extended to orthogonal multiple access (OMA) and FD systems with minor modifications as in \cite{arxiv_PIMRC}.
\subsection{System Model}
We consider user scheduling in a single-cell with $n$ users and one base-station (BS). Due to restricted computation complexity at network terminals, a limited number of users can be activated either in uplink (UL) or downlink (DL) at each time-slot. Subsets of users $\mathcal{V}_j, j\in [m]$ which can be activated simultaneously are called \textit{virtual users}, where $m\leq 2^n$. The set of all virtual users is denoted by $\mathsf{V}=\{\mathcal{V}_1,\mathcal{V}_2,\cdots,\mathcal{V}_m\}$. The choice of the active virtual user at a given time-slot determines the resulting system utility at that time-slot. The vector of system utilities due to activating each of the virtual users is called the \textit{performance vector}. For instance, the performance vector can be taken as the vector of sum-rates resulting from activating each of the virtual users. In this case, the performance vector is random and its value depends on the realization of the underlying  time-varying channel. In a given time-slot, the system utilities due to activating different virtual users may depend on each other. However, the performance vectors in different time-slots are assumed to be independent of each other, for example due to independence of the channels.
\begin{Definition}[\bf{Performance Vector}]
The vector of jointly continuous variables $(R_{1,t},R_{2,t}, \cdots, R_{m,t}), t\in \mathbb{N}$ is the performance vector of the virtual users at time $t$. The sequence $(R_{1,t},R_{2,t}, \cdots, R_{m,t}), t\in \mathbb{N}$ is a sequence of independent vectors distributed identically  according to the joint density $f_{R^m}$.
\end{Definition}

\begin{Example}[{\bf Parameters of the Performance Vector}]
\label{Ex:Ray}
Consider an OMA downlink scenario where only a single user may be activated at each time-slot. There are $m=n$ virtual users in this case, where $\mathcal{V}_i=\{u_i\}, i\in [n]$. The system utility is defined as the network throughput. Let $H_{i,t}=\beta_{i} G_{i,t}$ be the propagation channel coefficient between user $u_i$ and the BS at time-slot $t$. In this model, $\beta_{i}$ captures large-scale channel variations such as distance-dependent path loss and shadowing which mainly depend on the location of the user $u_i$. Furthermore, $G_{i,t}$ captures the small-scale variations of the channel caused by multi-path fading which depends on the scattering profile of the propagation environment. We assume that $\beta_{i}$ is constant over the time interval of interest and $G_{i,t}, i\in [n], t\in \mathbb{N}$ follows a complex normal distribution as in Rayleigh fading model. Additionally, it is assumed that $G_{i,t}, i\in [n]$ are independent over time. Consequently, channel coefficients $H_{i,t}, i\in [n]$ are also independent over time as $\beta_{i}, i\in [n]$ are constant. The resulting signal to noise ratio (SNR) of user $u_i$ at time-slot $t$ is defined as $\text{SNR}_{i,t}={p|H_{i,t}|^2}/{\sigma^2}$ where, $p$ and $\sigma^2$ denote the downlink transmit power and noise power, respectively. Consequently, the performance value for virtual user $\mathcal{V}_i$ at time-slot $t$ is defined by $R_{i,t}=\max\{\log_2(1+\text{SNR}_{i,t}),\gamma_{max}\}$, where $\gamma_{max}$ models the maximum spectral efficiency in the system. As a result, the PDF of the performance vector $R^m$, i.e. $f_{R^m}$, depends on the distribution of the underlying propagation channels which themselves depend on a set of model parameters such as the users' locations. The performance vector in NOMA systems can be parametrized similarly. Generally, we assume that the statistics of the performance vector is parametrized by some fixed model parameters $(\theta_1,\theta_2,\cdots,\theta_n)$, so that the PDF of $R^m$ is written as $f_{R^m}(r^m;\theta_1,\theta_2,\cdots,\theta_n), r^m\in \mathbb{R}^m$. 
\end{Example}

Temporal fairness requires that the fraction of time-slots in which each user is activated be bounded from below (above). The vector of lower (upper) bounds $\underline{w}^n$ ($\overline{w}^n)$ is called the lower (upper) temporal demand vector. The objective is to design a scheduling strategy satisfying the temporal fairness constraints while maximizing the resulting system utility. Accordingly, a scheduling strategy is defined as follows.

\begin{Definition}[\bf{Scheduler}] \label{Def:Strategy}
Consider the scheduling setup parametrized by $(n,\mathsf{V}, \underline{w}^n, \overline{w}^n, f_{R^m})$. A scheduling strategy $Q= (Q_t)_{t\in \mathbb{N}}$ is a family of (possibly stochastic) functions $Q_t: \mathbb{R}^{m\times t} \to \mathsf{V}, t\in \mathbb{N}$, where:
\begin{itemize}[leftmargin=*]
    \item { The input to $Q_t, t\in \mathbb{N}$ is the matrix of performance vectors $R^{m\times t}$ which consists of $t$ independently and identically distributed column vectors with distribution} $f_{R^m}$.% The vector $(R_{1,k}, R_{2,k},\cdots, R_{m,k}), k\in [t]$ is an independent realization of the performance vector $R^m$ corresponding to time-slot $k$. 
\item{ The temporal demand constraints are satisfied:
\begin{align}
P\left(\underline{w}_i\leq \underline{A}_{i}^{Q}\quad \& \quad \overline{A}^Q_i \leq \overline{w}_i, i\in [n]\right)=1,
\label{Def:tem_fair}
\end{align}}
\end{itemize}
where, the temporal share of user $u_i, i\in [n]$ up to time $t\in [s]$ is defined as
\begin{align}
&A^Q_{i,t}=\frac{1}{t}\sum^t_{k=1}\mathbbm{1}_{\big\{u_i\in Q_k(R^{m\times k})\big\}}, \forall i\in [n], t\in \mathbb{N},\\
&\underline{A}^Q_{i}= \liminf_{t\to \infty}A^Q_{i,t}, \qquad 
\overline{A}^Q_{i}= \limsup_{t\to \infty}A^Q_{i,t}
\label{Def:temp_share}
\end{align}
\end{Definition}

We consider homogeneous systems where the scheduler is allowed to activate subsets of at most $N_{max}$ users at each time-slot.  More precisely, for a homogeneous multi-user system with $n$ users and maximum number of active users $N_{max}\leq n$, the set of virtual users is defined as
\[\mathsf{V}= \left\{\mathcal{V}_j\subset \mathcal{U}\big| |\mathcal{V}_j|\leq N_{max}\right\}. \]
We write $(n,N_{max}, \underline{w}^n,$ $\overline{w}^n,f_{R^m})$ instead of $(n,\mathsf{V}, \underline{w}^n,$ $\overline{w}^n,f_{R^m})$ to characterize a 
homogeneous system.
A scheduling setup where the user temporal shares are required to take a specific value, i.e. ${A}_{i,s}^Q=w_i, i\in [n]$, is called a setup with \textit{equality temporal constraints} and is parametrized by $(n,N_{max}, {w}^n, {w}^n, f_{R^m})$. The average system utility of a scheduler is defined as:

\begin{Definition}[\bf{System Utility}]
For an $s$-scheduler $Q$:
\begin{itemize}[wide=0pt]
    \item {The average system utility up to time t, is defined as 
\begin{align}
U^Q_t&=\frac{1}{t}\sum^t_{k=1}\sum_{j=1}^m R_{j,k}\mathbbm{1}_{\big\{Q_k(R^{m\times k})=\mathcal{V}_{j}\big\}}.
\label{Def:sys_utility}
\end{align}}
\item{The variable $U^Q$ is called the average system utility for the scheduler, where
\begin{align*}
    U^Q=\liminf_{t\to \infty} U_t^Q
\end{align*}}
A scheduler $Q^*$ is optimal if and only if 
$Q^*\in\argmax_{Q\in \mathcal{Q}} U^Q$,
where $\mathcal{Q}$ is the set of all temporally fair schedulers for the scheduling setup. The optimal utility is denoted by $U^*$.
\end{itemize}
\end{Definition}
\subsection{Prior Literature} \label{subsec:lit}
In \cite{shahsavari2019general}, we showed that a class of scheduling strategies called TBSs achieve the optimal system utility subject to temporal fairness constraints. A TBS is formally defined below.

\begin{Definition}[\bf{TBS}]
For the scheduling setup $(n,N_{max}, \underline{w}^n, \overline{w}^n, f_{R^m})$ a threshold based strategy (TBS) is characterized by the vector $\lambda^n\in \mathbb{R}^n$. The strategy $Q_{TBS}(\lambda^n)=(Q_{TBS,t})_{t\in \mathbb{N}}$ is defined as:
\begin{align}
Q_{TBS,t}\big(R^{m\times t}\big)=\argmax_{\mathcal{V}_j\in\mathsf{V}} ~R_{j,t}+\sum_{i=1}^n \lambda_i \mathbbm{1}_{\{u_i\in\mathcal{V}_{j}\}}, ~t\in \mathbb{N}.
\label{Eq:thresh_str}
\end{align}
 The resulting temporal shares are represented as $A_i^{Q_{TBS}}, i\in [n]$. The utility of the TBS is written as $U_{w^n}(\lambda^n)$. 
The space of all threshold based strategies is denoted by $\mathcal{Q}_{TBS}$.
\label{def:U_TBS}
\end{Definition}

It can be shown that any optimal scheduling strategy can be represented as a TBS. In other words, any optimal strategy is \textit{equivalent} to a TBS, where equivalence is defined in \cite{shahsavari2019general}.

\begin{Theorem}[\!\!\cite{shahsavari2019general}]
For the scheduling setup $(n,N_{max}, \underline{w}^n,$ $ \overline{w}^n, f_{R^m})$, assume that $\mathcal{Q}\neq \varnothing$. Then, there exists an optimal threshold based strategy $Q_{TBS}$.
\label{th:neq:normal}
\end{Theorem}
%In Theorem \ref{th:neq:normal}, the equivalence relation $Q\sim Q'$ roughly implies that at each time-slot, the outputs of schedulers $Q$ and $Q'$ are equal with probability one. 
The theorem proves the existence of optimal TBSs. However, the question of how to construct such TBSs is not addressed. An iterative algorithm based on the Robins-Monro method was proposed in \cite{shahsavari2019general} to construct the optimal thresholds using the BS's observations of the channel realizations at each time-slot. It was shown that the output of the algorithm converges to the optimal thresholds under mild assumptions on the channel statistics. In this paper, we are interested in providing upper and lower bounds on the RoC of the thresholds and the effect on average system utility. 

\subsection{Rates of Convergence}
The optimal threshold vector ${\lambda^*}^n$ is a function of the channel statistics. More precisely, it can be shown that for the scheduling problem with equality constraints $(n,N_{max}, {w}^n, {w}^n, f_{R^m})$, the optimal threshold vector ${\lambda^*}^n$ is the unique vector for which the following fairness constraints hold:
\begin{align}
  &  P\left(\max_{\mathcal{V}_j:u_i\in \mathcal{V}_j}R_{j,t}+\sum_{i=1}^n \lambda_i \mathbbm{1}_{\{u_i\in\mathcal{V}_{j}\}}\geq \max_{\mathcal{V}_j:u_i\notin \mathcal{V}_j}R_{j,t}+\sum_{i=1}^n \lambda_i \mathbbm{1}_{\{u_i\in\mathcal{V}_{j}\}}\right)\nonumber
    \\&\qquad\qquad\qquad\qquad\qquad\qquad\qquad\qquad =w_i, i\in [n],
    \label{eq:thresh}
\end{align}
where the probability is taken with respect to $f_{R^m}$. 

If the BS has access to the distribution $f_{R^m}$, it may solve Equation \eqref{eq:thresh} to derive the optimal threshold vector. However, in practice, the BS does not have access to the statistics of the performance vector ${R^m}$. Rather, at time $t$, it estimates the optimal threshold vector using the realizations  ${R}^{m\times t}=r^{m\times t}$. Let $\widehat{\lambda}^n_t$ be the estimate of the optimal threshold vector at time $t$. In this work, we are interested in deriving upper and lower bound on the RoC of the sequence of vectors $\widehat{\lambda}^n_t$ to ${\lambda^*}^n$, where, we consider the RoC under equality temporal fairness constraints. The analysis can be extended to the case of inequality constraints in a straightforward fashion.
\begin{Definition}[\textbf{Threshold RoC}]
\label{def:th_RoC}
Consider the set of scheduling setups with equality constraints $(n,N_{max}, {w}^n, {w}^n, f_{R^m}), f_{R^m}\in \mathcal{P}$, where $\mathcal{P}$ is a set of PDFs. Define the space of all mapping from $\mathbb{R}^{m\times t}$ to $\mathbb{R}^n$ as $\mathcal{G}_t=\{g:\mathbb{R}^{m\times t}\to \mathbb{R}^n\}$. The optimal threshold RoC is defined as
\begin{align*}
    \alpha^*= \sup_{(g_t)_{t\in \mathbb{N}}\in 
    \prod_{t\in \mathbb{N}}\mathcal{G}_t}\inf_{f_{R^m}\in \mathcal{P}} \sup_{\alpha\geq 0}\bigg\{\alpha: \lim_{t \to \infty} 
    \frac{\mathbb{E}_{R^{m\times t}}(||{\lambda^*}^n-\widehat{\lambda}^n_t||^2_2)}{t^{-2\alpha}}<\infty
    \bigg\},
\end{align*}
 where $\widehat{\lambda}^n_t\triangleq g_t(R^{m\times t})$ and $||\cdot||_2$ is the $\ell_2$ norm. 
\end{Definition}
The infimum in the above definition is taken to derive the worst-case RoC over all possible channel statistics.
The long-term-fair utility RoC (LTU-RoC) is defined as:

\begin{Definition}[\textbf{LTU-RoC}]
Consider the set of scheduling setups with equality constraints $(n,N_{max}, {w}^n, {w}^n, f_{R^m}), f_{R^m}\in \mathcal{P}$, where $\mathcal{P}$ is a set of PDFs. Define the space of all mappings from $\mathbb{R}^{m\times t}$ to $[m]$ as $\mathcal{Q}_t=\{Q_t:\mathbb{R}^{m\times t}\to [m]\}$. The optimal LTU-RoC is defined as
\begin{align*}
    \zeta^*=& \sup_{(Q_t)_{t\in\mathbb{N}}\subset \prod_{t\in \mathbb{N}}\mathcal{Q}_t} \inf_{f_{R^m}\in \mathcal{P}} \\&\qquad\sup_{\zeta\geq 0}\{\zeta: \lim_{t \to \infty} 
    \frac{\mathbb{E}_{R^{m\times t}}(|{U^*}-\frac{1}{t}\sum_{i\in [t]}\widehat{U}_i|_+)}{t^{-\zeta}}<\infty\},
\end{align*}
where $\widehat{U}_t$ is the utility of $Q_t$ and $|x|_+=x\cdot \mathbbm{1}_{x>0}$. 
\end{Definition}
Note that in the above definition $Q_t$ need not be a TBS. 

\section{Bounds on the Rates of Convergence}
\label{sec:thresh}
We show that the threshold RoC is at most $\frac{1}{2}$, so that the threshold vector is at least\footnote{We write $f(x)=\omega(g(x))$ if $\lim_{n\to \infty}\frac{g(x)}{f(x)}< \infty$.} $\omega(\frac{1}{\sqrt{t}})$ away from the optimal value, where $t$ is the number of channel realizations.  In the case of LTU-RoC, we show that $\zeta^*$ can be arbitrarily large. To elaborate, we show that there are temporal fair schedulers for which $\frac{1}{t}\sum_{i\in [t]}\widehat{U}_i$, the average utility up to time $t$, may be larger than ${U^*}$ for all $t\in \mathbb{N}$  and approaches it from above.

%In the next section, we show through several simulations of practical scenarios, that the Robbins Monro based iterative algorithm for finding the optimal thresholds proposed in \cite{shahsavari2019general} achieves the optimal RoC of $\frac{1}{2}$. We further provide simulations which evaluate the RoC of the resulting utility due the construction algorithm

\subsection{Threshold Rate of Convergence}
As a first step, we investigate the RoC of the best estimate of the optimal threshold vector, when the BS has access to $t\in \mathbb{N}$ prior observations of the users' channel realizations. We derive an upper-bound on the RoC assuming the observations of channel realizations are noiseless (i.e. perfect channel state information). It is straightforward to argue that the upper-bound holds for the case of noisy observations as well. The arguments presented in this section build upon the following extension of the Cram\'{e}r-Rao bound for parameter estimation. 
\begin{Lemma}[\!\!\cite{miller1978modified}]
Let $X^{m\times t}$ be a matrix of random variables with distribution $f_{X^{m\times t}}(X^{m\times t}; \theta, \alpha)$, where $\alpha$ and $\theta$ are deterministic model parameters and $f_{X^{m\times t}}(X^{m\times t}; \theta, \alpha)$ is twice differentiable with respect to $\alpha$ for any fixed $\theta$. Then, the variance of any unbiased estimator $\widehat{\alpha}(X^{m\times t})$ of $\alpha$ is bounded as follows:
\begin{align}
    \sigma^{2}_{\widehat{\alpha}}\geq \mathbb{E}_{\theta,X^{n\times t}} ^{-1}\left(\frac{\partial^2}{\partial \alpha^2} \frac{1}{ln\left(f_{X^{n\times t}}(X^{n\times t}; \theta,\alpha)\right)}\right).
    \label{Eq:CR}
\end{align}
\label{lem:CR}
\end{Lemma}

The right hand side of Equation \eqref{Eq:CR} is called the Fisher information of the variable $\alpha$.  

The optimal threshold vector ${\lambda^*}^n$ is a function of the PDF of the performance vector $R^m$. As in Example \ref{Ex:Ray}, the PDF of the performance vector is parameterized by a set of model parameters such as the users' locations. Generally, we assume that the PDF is parametrized by the real-valued vector, $(\theta_1,\theta_2,\cdots,\theta_n)$, so that the statistics of $R^m$ can be expressed as $f_{R^m}(R^m; \theta_1,\theta_2,\cdots, \theta_n)$. In order to derive an upper-bound on the threshold RoC, we consider a genie-assisted BS which is given the values of the model parameters $\theta_i, i\neq k$ and optimal thresholds $\lambda^*_i, i\neq k$, for some $k\in [n]$. Consequently, the genie-assisted BS does not know the model parameter $\theta_{k}$ and the optimal threshold $\lambda_{k}^*$.
Furthermore, we assume that the BS can accurately calculate $f_{R^m}(R^m; \theta_1,\theta_2,\cdots, \theta_n)$ and hence ${\lambda^*}^n$ provided that it has access to $\theta_1,\theta_2,\cdots, \theta_n$. Under these assumptions, the problem of finding the optimal threshold vector is related to the well-studied quantile estimation problem \cite{takeuchi2006nonparametric}. To elaborate, note that from Equation \eqref{eq:thresh} finding the optimal $\lambda^*_{k}$ requires solving the following equation:
\begin{align}
     \!\!\!w_{k}\!\!=\!P\left(\!\!\max_{\mathcal{V}_j:u_{k}\in \mathcal{V}_j}
         (R_{j,t}+\!\!\!\!\!\!\!\sum_{u_i\in \mathcal{V}_j- \{u_{k}\}}\!\!\!\!\!\!\!\lambda^*_{i})\!
         - \!\!
         \max_{\mathcal{V}_j:u_{k}\notin \mathcal{V}_j} (R_{j,t}+\!\!\sum_{u_i\in \mathcal{V}_j}\!\!\lambda^*_{i})
         )\geq\! -\lambda^*_{k}\right).
         \label{eq:t1}
\end{align}
Define $\widetilde{R}\triangleq \max_{\mathcal{V}_j:u_{k}\in \mathcal{V}_j}
         (R_{j,t}+\sum_{u_i\in \mathcal{V}_j- \{u_{k}\}}\lambda^*_{i})
         - 
         \max_{\mathcal{V}_j:u_{k}\notin \mathcal{V}_j} (R_{j,t}+\sum_{u_i\in \mathcal{V}_j}\lambda^*_{i})
         )$ 
         and let $f_{\widetilde{R}}(\widetilde{R}; \theta_1,\theta_2,\cdots,\theta_n)$ be the underlying probability measure. Equation \eqref{eq:t1} can be written in the form of the following integral equation which describes a quantile estimation problem:
         \begin{align}
            w_{k}=\int_{-\lambda^*_{k}}^{\infty}
            f_{\widetilde{R}}(r; \theta_1,\theta_2,\cdots,\theta_n)dr.
         \end{align}
    Consequently, $\lambda^*_{k}$ is a function of $\theta_{k}$ given $\theta_i, i\neq k$. We assume that $\lambda^*_{k}(\theta_1)$ is a smooth, differentiable function. 
    For a given $\theta_{k}=\theta$, assume that $\frac{\partial}{\partial \theta_{k}}\lambda^*_{k}|_{\theta_{k}=\theta}>0$, without loss of generality. Then, it is well-known that there exists a local neighborhood $\mathcal{B}=[\theta-\epsilon, \theta+\epsilon], \epsilon>0$, for which $\lambda^*_{k}(\theta_{k}), \theta_{k}\in \mathcal{B}$ is increasing in $\theta_1$. Particularly, $\lambda^*_{k}$ is one-to-one as a function of $\theta_{k}$ over the interval $\mathcal{B}$. As a result, the PDF of $\widetilde{R}$ can be parametrized by $\theta_1,\theta_2,\cdots,\theta_{k-1},\lambda_{k}^*, \theta_{k+1},\cdots ,\theta_n$ in this neighborhood. We assume that $f_{\tilde{R}}(\widetilde{R};\theta_1,\theta_2,\cdots,\theta_{k-1},\lambda_{k}^*, \theta_{k+1},\cdots ,\theta_n)$ is twice differentiable in $\lambda_{k}^*$. Note that this property can be verified in conventional models of cellular communication systems such as the one described in Example \ref{Ex:Ray}. 
    
    The BS needs to estimate the parameter $\lambda^*_{k}$ using the observations $R^{m\times t}$. This resembles the problem described in Lemma \ref{lem:CR}. Using the Cram\'er-Rao bound for parameter estimation, we prove the following theorem which provides upper-bounds on the threshold RoC when unbiased estimators are used. 

\begin{Theorem}
\label{th:th_RoC}
Let $\sigma^*_{k}$ denote the minimum mean square error (MMSE) in estimating $\lambda^*_{k}$ for the genie-assisted setup described above. The following holds $
    \sigma^{2}_{k}\geq I_{\widetilde{R}}(\lambda^*_{k})$,
where 
\begin{align*}
I_{\widetilde{R}}(\lambda^*_{k})= \frac{1}{t}\mathbb{E}_{\widetilde{R}} ^{-1}\left(\frac{\partial^2}{\partial {\lambda^*_{k}}^2} {ln\left(f_{\tilde{R}}(\widetilde{R};\underline{\theta}\right)}\right),
\end{align*}
where $\underline{\theta}=(\theta_1,\theta_2,\cdots,\theta_{k-1},\lambda_{k}^*, \theta_{k+1},\cdots ,\theta_n)$. Particularly, we have $\alpha^*\leq \frac{1}{2}$, where $\alpha^*$ is the threshold RoC.
\end{Theorem}
\begin{proof}
Appendix \ref{Ap:th:th_Roc}.
\end{proof}
      \begin{figure*}[t]
    \centering
        \begin{subfigure}[b]{0.31\textwidth}
            \includegraphics[width=\textwidth]{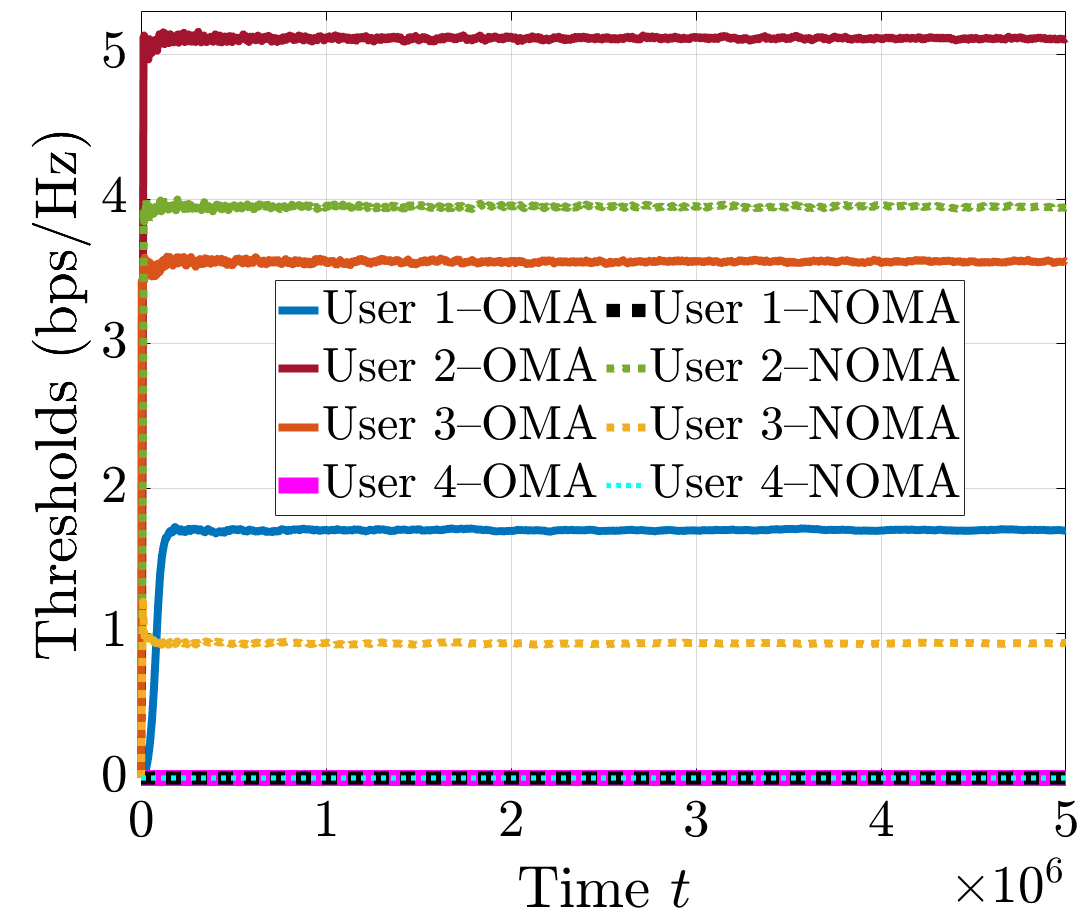}
            \caption{Thresholds}%
            {{\small }}    
            \label{fig:thresholds}
        \end{subfigure}
        \begin{subfigure}[b]{0.31\textwidth}              \includegraphics[width=\textwidth]{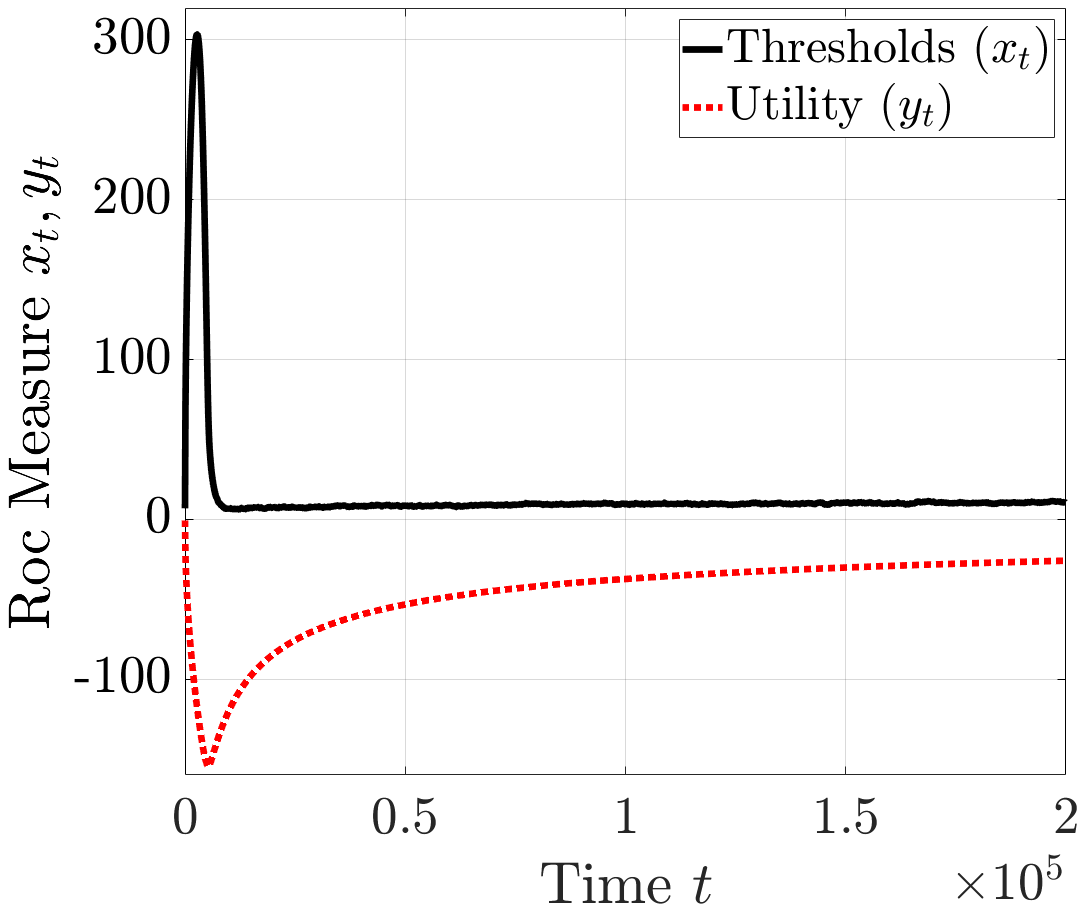}
            \caption{RoC in OMA}%
            {{\small}}    
            \label{fig:roc-oma}
        \end{subfigure}
        \begin{subfigure}[b]{0.31\textwidth}   
            \centering 
            \includegraphics[width=\textwidth]{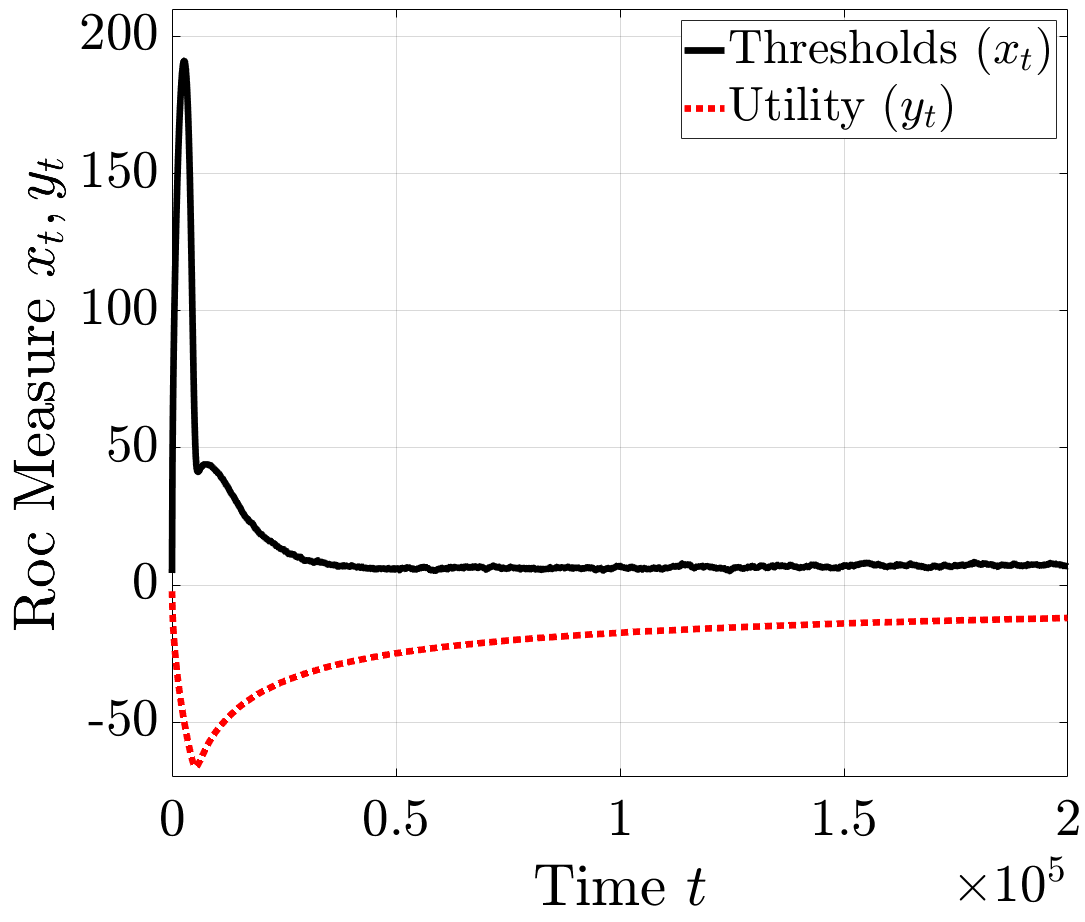}
            \caption{RoC in NOMA}%
            {{\small}}
            \label{fig:roc-noma}
        \end{subfigure}

        \caption{ (a) The evolution of the user thresholds in time for OMA and NOMA settings. RoC measures $x_t$ (thresholds) and $y_t$ (utility) as a function of time for the thresholds and system utility in the (b) OMA setting and (c) NOMA setting. %The horizontal axis in each plot is the sampled time-slot index, where the sampling parameter $H$ is set to $10^{-2}$.
        }
        \label{fig:fig}
        %  \vspace{-4 mm}
\end{figure*}

Theorem \ref{th:th_RoC} considers the threshold RoC when unbiased estimators are used to estimate $\lambda^*_{k}$. Similar bounds can potentially be derived when biased estimators are employed
using an alternative median estimation argument based on the order statistics of the vectors $R^{m}(1),R^{m}(2),\cdots,R^{m}(t)$
as in \cite{hojo1931distribution}.

\subsection{Utility Rate of Convergence}
%So far, we have derived an upper-bound on the threshold RoC.
In the following, we show that the LTU-RoC is arbitrarily large. In fact,  it is shown that there exists a scheduler construction strategy 
for which
%which produces the scheduler $(Q_t)_{t\in\mathbb{N}}$ using the empirical observations of the users' channels for which $|{U^*}-\frac{1}{t}\sum_{i\in [t]}\widehat{U}_i|_+=0, \forall t\in \mathbb{N}$. So that 
the utility of the scheduler approaches the long-term optimal utility from above as the scheduling window increases asymptotically. 

         \begin{Theorem}
         \label{th:zeta}
Consider the set of scheduling setups  $(n,N_{max}, \underline{w}^n, \overline{w}^n, f_{R^m}), f_{R^m} \in \mathcal{P}$, where $\mathcal{P}$ is a compact set of bounded and differentiable PDFs on $\mathcal{R}^m$. Assume that $\alpha^*>0$ for this set of scheduling setups. There exists a scheduling strategy $Q^*=(Q^*_t)_{t\in \mathbb{N}}$ for which $|{U^*}-\frac{1}{t}\sum_{i\in [t]}\widehat{U}_i|_+=0$ for all large enough $t\in \mathbb{N}$. Particularly, we have
$\zeta^*=\infty$.
\end{Theorem}
\textit{Proof.} Appendix \ref{Ap:th:zeta}.
\section{Simulation Results}
\label{Sec:sim}

% \begin{figure}
%      \centering
%      \begin{subfigure}[h]{0.75\linewidth}
%          \centering 
%              \includegraphics[width=\linewidth]{Figures/threshold-oma.png}
%              \caption{OMA setting}
%              \label{fig:threshold-oma}
%      \end{subfigure}
%      \vfill \vspace{3pt}
%      \begin{subfigure}[h]{0.75\linewidth}
%         \centering 
%              \includegraphics[width=\linewidth]{Figures/threshold-noma.png}
%              \caption{NOMA setting}
%              \label{fig:threshold-noma}
%      \end{subfigure}
%         \caption{The evolution of the user thresholds in time for (a) OMA setting and (b) NOMA setting. The horizontal axis is the sampled time-slot index, where the sampling parameter $H$ is set to $10^{-2}$.}
%         \label{fig:threshold}
% \end{figure}

In this section, we provide simulation results to investigate the RoC of the Robbins Monro based iterative algorithm for finding the optimal thresholds of TBS, which was proposed in \cite[Algorithm 2]{shahsavari2019general}. This algorithm starts from zero thresholds for all users and updates the thresholds based on user temporal demands as well as the the outcome of the scheduling at each time-slot. We refer to this algorithm as the \textit{threshold learning algorithm} (TLA). We consider a time-slotted single small-cell scenario with a BS located in the center and four users distributed uniformly at random in a ring around the BS with inner and outer radii of 20 m and 100 m, respectively. We investigate two communication settings. In the first setting (OMA), only one user is scheduled at each time-slot, i.e. $N_{max}=1$. In the second setting (NOMA) one user or a pair of users are scheduled at each time-slot, i.e. $N_{max}=2$. %Additionally, we only consider lower temporal demand for each user and assume that there are no upper temporal demand constraints. 
The network utility is modeled by the truncated Shannon sum-rate as in \cite{shahsavari2019general} with maximum allowed spectral efficiency of $\gamma_{max}=6$ bps/Hz. At each time-slot, prior to scheduling, a max-min power optimization is performed for each virtual user as in \cite{liu2016fairness}. 
% For a given virtual user, we find the transmit power which maximizes the minimum individual user rates in that virtual user. This max-min optimization allows for a balanced rate allocation within the virtual user. It can be shown that the max-min optimization is quasi-concave. Consequently, quasi-concave programming methods such as bisection search can be used to find the optimal transmit powers \cite{liu2016fairness}.
Maximum BS transmit power constraint is chosen such that the average SNR of $10$ dB is achievable when a single user is active on the boundary of the cell. In each setting, we use TLA to obtain an estimate of the optimal user thresholds. The total number of time-slots is set to $5\times 10^6$. According to \cite [Section V-A]{shahsavari2019general}, the step-size $s_t$ used to update the thresholds at time-slot $t$ should satisfy the following constraints: i) $s_t>0$, ii) $\lim_{t\rightarrow \infty} s_t=0$, and iii) $\sum_{t=1}^\infty s_t=\infty$, $\sum_{t=1}^\infty s^2_t<\infty$. We take the step-size to be $s_t=t^{-0.7}$ which satisfies these conditions. 

We investigate the convergence of the thresholds in OMA and NOMA settings when using TBS along with TLA. We consider a random user distribution and assume that the lower temporal demand vector is $\underline{w}^4=[0.1,0.2,0.3,0.4]$ which can be shown to be feasible for both communication settings as $\sum_i \underline{w}_i \leq 1$ \cite [Section IV-B]{shahsavari2019general}. Figure \ref{fig:fig}(a) depicts the time-evolution of the user thresholds when using TLA. We observe that the thresholds converge in both settings.

% \begin{figure}
%      \centering
%      \begin{subfigure}[b]{0.8\linewidth}
%          \centering 
%              \includegraphics[width=\linewidth]{Figures/Roc-oma-v2.png}
%              \caption{OMA setting}
%              \label{fig:roc-oma}
%      \end{subfigure}
%      \vfill \vspace{3pt}
%      \begin{subfigure}[b]{0.8\linewidth}
%         \centering 
%              \includegraphics[width=\linewidth]{Figures/Roc-noma-v2.png}
%              \caption{NOMA setting}
%              \label{fig:roc-noma}
%      \end{subfigure}
%         \caption{RoC measures $x_t$ (thresholds) and $y_t$ (utility) as a function of time for the thresholds and system utility in the (a) OMA setting and (b) NOMA setting. The horizontal axis is the sampled time-slot index, where the sampling parameter $H$ is set to $10^{-2}$.}
%         \label{fig:roc}
% \end{figure}

Next, we investigate the RoC for the thresholds and system utility in OMA and NOMA settings in Figures \ref{fig:fig}(b) and (c). We define \textit{RoC measure} at time-slot $t$, as $x_t=t^{1/2}\mathbb{E}_{R^{m\times t}}(||{\lambda^*}^n-\widehat{\lambda}^n_t||_{\infty})$ and $y_t=t^{1/2}\mathbb{E}_{R^{m\times t}}({U^*}-\widehat{U}_t)$ for the thresholds and the system utility, respectively. We note that if $x_t$ converges to a constant value, it indicates that the thresholds approach the optimal value as $\omega(\frac{1}{\sqrt(t)})$ when using TLA. According to Theorem \ref{th:zeta}, the RoC for system utility can be infinity if $\exists t': U_t>U^*,~t\geq t'$. Therefore, we do not include operator $|.|_+$ in $y_t$ to identify if $U_t$ converges $U^*$ from above. We approximate the expectation by empirical average over 100 channel realizations. Estimates of the optimal thresholds ${\lambda^*}^n$ and system utility $U^*$ are obtained by running TLA for $5\times10^6$ time-slots. Figures \ref{fig:fig}(b) and (c) depict $x_t$ and $y_t$ for the first $2\times10^5$ time-slots %where the sampling parameter $H$ (applied to x-axis), is set to $10^{-2}$, i.e. one sample per $10^2$ time-slots. 
We observe that the sequence $x_t, t\in \mathbb{N}$ converges to a constant value in both OMA and NOMA settings, which is consistent with the result of Theorem \ref{th:th_RoC}. Furthermore, Figure 1 suggests that thresholds approach the optimum ones as $\omega\Big(\frac{1}{\sqrt{t}}\Big)$, leading the largest RoC of  $\alpha^*=\frac{1}{2}$ according to Theorem \ref{th:th_RoC}. On the other hand, we observe that $y_t$ is negative implying that sequence $\widehat{U}_t$ converges to $U^*$ from above when using TLA (note that the convergence of utility and thresholds is guaranteed according to Section \ref{subsec:lit}). The reason is that TLA starts from zero thresholds leading to a pure opportunistic scheduler at the beginning. Therefore, virtual users are initially chosen based on their performance value leading to a system utility higher than the ultimate system utility $U^*$ when temporal demands are satisfied. %However, the thresholds converge to their optimal values in time so as to satisfy the temporal demands. As a results, virtual users with lower performance values are also chosen whenever necessary to meet the demands. Hence, the system utility converges to $U^*$ from above. 
This is in agreement with Theorem \ref{th:zeta}.

\section{Conclusion}
We have considered multi-user scheduling under temporal fairness constraints. We have investigated the rates of convergence of the best mean square estimates of the scheduling threshold vector, when threshold based strategies are used. We have shown that the rate of convergence is at most $\frac{1}{2}$, so that the best threshold vector estimate is at least $\omega(\frac{1}{\sqrt{t}})$ away from the optimal threshold vector, where $t$ is the number of instances of users' channel realizations used in estimation. We have further show that the utility under long-term fairness constraints may approach the optimal utility from above. We have provided simulations of practical scenarios verifying our results. 
\label{Sec:Conc}

\begin{appendices}
\section{Proof of Theorem \ref{th:th_RoC}}
\label{Ap:th:th_Roc}
From Lemma \ref{lem:CR}, we have:
\begin{align*}
    \mathbb{E}_{R^{m\times t}}(||{\lambda^*}^n-\widehat{\lambda}^n_t||^2_2)\geq
    \mathbb{E}_{\tilde{R}^t} ^{-1}\left(\frac{\partial^2}{\partial {\lambda^*_{i'}}^2} {ln\left(f_{\tilde{R}^{t}}(\tilde{R}^{ t}; \underline{\theta})\right)}\right)
    =\frac{1}{t}\mathbb{E}_{\tilde{R}} ^{-1}\left(\frac{\partial^2}{\partial {\lambda^*_{i'}}^2} {ln\left(f_{\tilde{R}}(\tilde{R}; \underline{\theta})\right)}\right),
\end{align*}
where in the last equality we have used the fact that the channel realizations at different time-slots are independent to conclude that $f_{\tilde{R}^{t}}(\tilde{R}^{ t}; \underline{\theta})=
\prod_{i\in [t]}f_{\tilde{R}}(\tilde{R}_i; \underline{\theta})$, and the fact that the channel realizations are identically distributed in different time-slots to write $\mathbb{E}_{\tilde{R}^{t}}\left(\frac{\partial^2}{\partial {\lambda^*_{i'}}^2} {ln\left(f_{\tilde{R}^{t}}(\tilde{R}^{ t}; \underline{\theta})\right)}\right)= t \mathbb{E}_{\tilde{R}}\left(\frac{\partial^2}{\partial {\lambda^*_{i'}}^2} {ln\left(f_{\tilde{R}}(\tilde{R}^{ t}; \underline{\theta})\right)}\right)$. 

Consequently, we have shown that $\mathbb{E}_{R^{m\times t}}(||{\lambda^*}^n-\widehat{\lambda}^n_t||^2_2)\geq \frac{c}{t}$ for some constant $c>0$. As a result, from Definition \ref{def:th_RoC}, we conclude that $\alpha^*\leq \frac{1}{2}$ since
\begin{align*}
    \lim_{t \to \infty} 
    \frac{\mathbb{E}_{R^{m\times t}}(||{\lambda^*}^n-\widehat{\lambda}^n_t||^2_2)}{t^{-2\alpha}}\geq \lim_{t \to \infty} 
    \frac{ct^{-1}}{t^{-2\alpha}}= \infty, \quad \forall \alpha>\frac{1}{2}. 
\end{align*}
\section{Proof of Theorem \ref{th:zeta}}
\label{Ap:th:zeta}
From the theorem statement, the threshold RoC is equal to $\alpha^*>0$. So, there exists a family of threshold construction functions $f_t:\mathbb{R}^{m\times t} \to \mathbb{R}^n$ and constant $c\in \mathbb{R}$ for which 
$\mathbb{E}_{R^{m\times t}}(||{\lambda^*}^n-\widehat{\lambda}^n_t||^2_2)\leq \frac{1}{ct^{2\alpha^*}}$ for any $f_{R^m}\in \mathcal{P}$ and large enough $t\in \mathbb{N}$, where $\widehat{\lambda}_t^n=f(R^{m\times t})$. 
Consequently, by the Chebyshev inequality, we have:
\begin{align}
P(||{\lambda^*}^n-\widehat{\lambda}^n_t||_2> \frac{\sqrt{c}}{t^{0.5\alpha^*}})\leq \frac{\mathbb{E}_{R^{m\times t}}(||{\lambda^*}^n-\widehat{\lambda}^n_t||^2_2)}{(\sqrt{c}t^{-0.5\alpha^*})^2}\leq  \frac{1}{t^{\alpha^*}}.   
\label{Eq:pf31}
\end{align}

Note that $c$ and $\alpha^* $ are  universal constants which do not depend on the underlying distribution $f_{R^m}$. We assume that the scheduler can calculate these constants. Let $V_t(\lambda^n)$ be the utility due to the TBS with threshold $\lambda^n\in \mathbb{R}^n$ at time $t\in \mathbb{N}$. By definition, we have $V_t(\lambda^n)=R_{J}$, where $J=\argmax_{\mathcal{V}_j\in\mathsf{V}} ~S_{\lambda^n}\big(\mathcal{V}_j,R_{j,t}\big)$, where $S_{\lambda^n}\big(\mathcal{V}_j,R_{j,t}\big)$ is the scheduling measure of the $j$th user when the threhsold vector is $\lambda^n$. Note that $S_{\lambda^n}\big(\mathcal{V}_j,R_{j,t}\big)- S_{{\lambda'}^n}\big(\mathcal{V}_j,R_{j,t}\big)
\leq n\times max_{i\in [n]} |\lambda_i-\lambda'_i|\leq
n\times ||\lambda^n-{\lambda'}^n||_2$, where $\lambda^n,{\lambda'}^n\in \mathbb{R}^n$. Consequently, from Eqaution \eqref{Eq:pf31}, we have:
\[
P(|V_t({\lambda^*}^n)- V_t(\widehat{\lambda}^n_t)|\leq \frac{n\sqrt{c}}{t^{0.5\alpha^*}})\geq 1-\frac{1}{t^{\alpha^*}}.
\]
It follows that $\mathbb{E}_{R^{m\times t}}(|{U}({\lambda^*}^n)-{U}({\widehat{\lambda}}_t^n)|)\leq \frac{c'}{t^{0.5\alpha^*}}$, where $U(\lambda^n)$ is the average system utility due to the TBS with threshold vector $\lambda^n$ and $c'>0$ is constant in $t$. Furthermore, $U(0^n)>U(\lambda^n), \forall \lambda^n\neq 0^n$ with probability one since the TBS with the all-zero threshold vector always activates the virtual user which gives the maximum system utility. Let $\mathbb{E}(U(0^n))-\mathbb{E}(U({\lambda^*}^n))=\epsilon>0$. Let us fix a natural number $M>2$. We describe the operation of the scheduler in the resource blocks $t\in [\sum_{i=1}^{k-1}M^i(1+M^{\frac{i\alpha^*}{4}})+1, \sum_{i=1}^{k}M^i(1+M^{\frac{i\alpha^*}{4}})]$ for any $k\in \mathbb{N}$. In the first $M^k$ resource blocks, the scheduler uses the TBS with threshold vector $0^n$ for scheduling and in the next $M^{k(1+\frac{\alpha^*}{4})}$ it uses the TBS with threshold vector $\widehat{\lambda}_{M^{k\alpha^*}}^n$. It is straightforward to show that the expected cumulative utility in the first $M^k$ resource blocks is $\mathbb{E}((U^*+\epsilon) M^k)$ and in the next $M^{k+\frac{k\alpha^*}{4}}$ it is at least $\mathbb{E}((U^*-\frac{c'}{M^{\frac{k\alpha^*}{2}}}) M^{k+\frac{k\alpha^*}{4}})$. So, we have 
\begin{align*}
    |\mathbb{E}(tU^*-\sum_{i\in [t]} \widehat{U}^n_i)|_+\leq 
    |-\sum_{i\in [k]}(\epsilon- M^{-\frac{i\alpha^*}{4}})M^i|_+,
\end{align*}
with probability one for asymptotically large $k$, where $t= \sum_{i\in [k]} M^i(1+M^{\frac{k\alpha^*}{4}})$. The right hand side is equal to $0$ for large enough $k$. This completes the proof. 
\end{appendices}

\newpage
\bibliographystyle{IEEEtran}
\bibliography{reference}

\begin{thebibliography}{10}
\providecommand{\url}[1]{#1}
\csname url@samestyle\endcsname
\providecommand{\newblock}{\relax}
\providecommand{\bibinfo}[2]{#2}
\providecommand{\BIBentrySTDinterwordspacing}{\spaceskip=0pt\relax}
\providecommand{\BIBentryALTinterwordstretchfactor}{4}
\providecommand{\BIBentryALTinterwordspacing}{\spaceskip=\fontdimen2\font plus
\BIBentryALTinterwordstretchfactor\fontdimen3\font minus
  \fontdimen4\font\relax}
\providecommand{\BIBforeignlanguage}[2]{{%
\expandafter\ifx\csname l@#1\endcsname\relax
\typeout{** WARNING: IEEEtran.bst: No hyphenation pattern has been}%
\typeout{** loaded for the language `#1'. Using the pattern for}%
\typeout{** the default language instead.}%
\else
\language=\csname l@#1\endcsname
\fi
#2}}
\providecommand{\BIBdecl}{\relax}
\BIBdecl

\bibitem{kulkarni2003opportunistic}
S.~S. Kulkarni and C.~Rosenberg, ``Opportunistic scheduling for wireless
  systems with multiple interfaces and multiple constraints,'' in \emph{Proc.
  ACM Intl. Workshop on Modeling Analysis and Simulation of Wireless and Mobile
  Systems}, 2003.

\bibitem{shahsavari2019general}
S.~{Shahsavari}, F.~{Shirani}, and E.~{Erkip}, ``A general framework for
  temporal fair user scheduling in {NOMA} systems,'' \emph{IEEE Journal of
  Selected Topics in Signal Processing}, vol.~13, no.~3, pp. 408--422, June
  2019.

\bibitem{liu-jsac}
X.~Liu, E.~K.~P. Chong, and N.~B. Shroff, ``Opportunistic transmission
  scheduling with resource-sharing constraints in wireless networks,''
  \emph{IEEE J. Sel. Areas Commun.}, vol.~19, no.~10, pp. 2053--2064, 2001.

\bibitem{issariyakul2004throughput}
T.~Issariyakul and E.~Hossain, ``Throughput and temporal fairness optimization
  in a multi-rate {TDMA} wireless network,'' in \emph{2004 IEEE International
  Conference on Communications}, vol.~7, 2004, pp. 4118--4122.

\bibitem{rosenberg-short-term}
S.~S. Kulkarni and C.~Rosenberg, ``Opportunistic scheduling policies for
  wireless systems with short term fairness constraints,'' in \emph{2003 IEEE
  Global Telecommunications Conference}, vol.~1, Dec 2003, pp. 533--537 Vol.1.

\bibitem{ahlswede1973multi}
R.~Ahlswede, ``Multi-way communication channels,'' in \emph{Second
  International Symposium on Information Theory: Tsahkadsor, Armenia, USSR,
  Sept. 2-8, 1971}, 1973.

\bibitem{yu2004sum}
W.~Yu and J.~M. Cioffi, ``Sum capacity of {G}aussian vector broadcast
  channels,'' \emph{IEEE Transactions on Information Theory}, vol.~50, no.~9,
  pp. 1875--1892, 2004.

\bibitem{arxiv_PIMRC}
S.~Shahsavari, A.~Khojastepour, F.~Shirani, and E.~Erkip, ``Opportunistic
  temporal fair mode selection and user scheduling for full-duplex systems,''
  \emph{available on arxiv.org}, 2019.

\bibitem{shiranishortterm}
S.~{Shahsavari}, F.~{Shirani}, and E.~{Erkip}, ``On the fundamental limits of
  multi-user scheduling under short-term fairness constraints,'' in \emph{2019
  IEEE International Symposium on Information Theory (ISIT)}, July, pp.
  2534--2538.

\bibitem{miller1978modified}
R.~Miller and C.~Chang, ``A modified {C}ram{\'e}r-{R}ao bound and its
  applications (corresp.),'' \emph{IEEE Transactions on Information theory},
  vol.~24, no.~3, pp. 398--400, 1978.

\bibitem{takeuchi2006nonparametric}
I.~Takeuchi, Q.~V. Le, T.~D. Sears, and A.~J. Smola, ``Nonparametric quantile
  estimation,'' \emph{Journal of Machine Learning Research}, vol.~7, no. Jul,
  pp. 1231--1264, 2006.

\bibitem{hojo1931distribution}
T.~Hojo and K.~Pearson, ``Distribution of the median, quartiles and
  interquartile distance in samples from a normal population,''
  \emph{Biometrika}, pp. 315--363, 1931.

\bibitem{liu2016fairness}
Y.~Liu, M.~Elkashlan, Z.~Ding, and G.~K. Karagiannidis, ``Fairness of user
  clustering in {MIMO} non-orthogonal multiple access systems,'' \emph{IEEE
  Communications Letters}, vol.~20, no.~7, pp. 1465--1468, 2016.

\end{thebibliography}

\end{document}